\documentclass[11pt,a4paper,reqno]{amsart}%
\usepackage{amsthm,amsmath,amsfonts,amssymb,amsxtra,appendix,bookmark,euscript,delarray,dsfont}
\usepackage[latin1]{inputenc}

\makeatletter
\def\@tocline#1#2#3#4#5#6#7{\relax
  \ifnum #1>\c@tocdepth 
  \else
    \par \addpenalty\@secpenalty\addvspace{#2}%
    \begingroup \hyphenpenalty\@M
    \@ifempty{#4}{%
      \@tempdima\csname r@tocindent\number#1\endcsname\relax
    }{%
      \@tempdima#4\relax
    }%
    \parindent\z@ \leftskip#3\relax \advance\leftskip\@tempdima\relax
    \rightskip\@pnumwidth plus4em \parfillskip-\@pnumwidth
    #5\leavevmode\hskip-\@tempdima
      \ifcase #1
       \or\or \hskip 1em \or \hskip 2em \else \hskip 3em \fi%
      #6\nobreak\relax
      \dotfill
      \hbox to\@pnumwidth{\@tocpagenum{#7}}
    \par
    \nobreak
    \endgroup
  \fi}
\makeatother

\theoremstyle{plain}
\newtheorem{theorem}{Theorem}

\newtheorem{corollary}[theorem]{Corollary}

\newtheorem{conjecture}[theorem]{Conjecture}
\theoremstyle{remark}

\newcommand\R{{\ensuremath {\mathbb R} }}
\newcommand\C{{\ensuremath {\mathbb C} }}

\newcommand\1{{\ensuremath {\mathds 1} }}

\newcommand\nn{\nonumber}
\newcommand\bq{\begin{equation}}
\newcommand\eq{\end{equation}}

\renewcommand\phi{\varphi}

\newcommand{\gH}{\mathfrak{H}}

\newcommand{\wto}{\rightharpoonup}

\newcommand{\cK}{\mathcal{K}}
\newcommand{\cE}{\mathcal{E}}

\newcommand{\cL}{\mathcal{L}}
\newcommand{\eps}{\epsilon}

\renewcommand{\epsilon}{\varepsilon}

\DeclareMathOperator{\Tr}{{\rm Tr}}

\renewcommand{\ge}{\geqslant}
\renewcommand{\le}{\leqslant}

\renewcommand{\leq}{\leqslant}

\title[Convergence of Levy-Lieb to Thomas-Fermi functional]{Convergence of Levy-Lieb to Thomas-Fermi density functional}

\author[N. Gottschling]{Nina Gottschling}
\address{Ludwig Maximilian University of Munich, Department of Mathematics, Theresienstrasse 39, D-80333 Munich, Germany}
\email{ngottschling@posteo.net,  nmg43@cam.ac.uk}

\author[P.~T. Nam]{Phan Th\`anh Nam}
\address{Ludwig Maximilian University of Munich, Department of Mathematics, Theresienstrasse 39, D-80333 Munich, Germany} 
\email{nam@math.lmu.de}

\begin{document}
\date{\today}

\begin{abstract} We prove that the Levy-Lieb density functional Gamma-converges to the Thomas-Fermi functional in the semiclassical mean-field limit. In particular, this aides an easy alternative proof of the validity of the atomic Thomas-Fermi theory which was first established by Lieb and Simon. 
\end{abstract}

\maketitle

\bigskip\bigskip

\bigskip\bigskip

\setcounter{tocdepth}{2}
\tableofcontents

\section{Introduction}

From first principles of quantum mechanics, the total energy of $N$ identical fermions in $\R^d$ with spin $q\ge 1$ can be described by a Hamiltonian $H_N$ in the Hilbert space
$$\gH^N_a= L^2_a\left((\R^{d}\times \{1,2,...,q\})^N;\C\right),$$
which contains wave functions which are anti-symmetric under the permutations of space-spin variables:
$$
\Psi_N(...,(x_i,\sigma_i),...,(x_j,\sigma_j),...)=-\Psi_N(...,(x_i,\sigma_i),...,(x_j,\sigma_j),...).
$$ 
In particular, the ground state energy of the system is  
\begin{align}\label{eq:EN-HN}
E_N^{\rm QM} &= \inf \left\{ \langle \Psi_N, H_N \Psi_N \rangle: \Psi_N \in S_N\right\},
\end{align}
where $S_N$ is the set of all (normalized) wave functions in the quadratic form domain of $H_N$. 

Although the above microscopic theory is very precise, it usually becomes too complicated for practical calculations when $N$ is large. Therefore, it is desirable to develop effective theories which depend on less variables but still capture some collective properties of the system in certain regimes. 

\subsection{Levy-Lieb and Thomas-Fermi density functionals}

In density functional theory, instead of considering a complicated wave function $\Psi_N\in S_N $ one simply looks at its one-body density 
$$
\rho_{\Psi_N}(x)= N \sum_{\sigma_1,...,\sigma_N\in \{1,...,q\}} \int_{(\R^d)^{N-1}} |\Psi_N((x,\sigma_1),(x_2,\sigma_2),...,(x_N,\sigma_N))|^2 d x_2 ... d x_N,
$$
which satisfies the simple constraints
$$\rho_{\Psi_N}(x)\ge 0, \quad \int_{\R^d} \rho_{\Psi_N}(x) d x =N. $$

The idea of describing a quantum state using its one-body density goes back to Thomas \cite{Thomas-27} and Fermi \cite{Fermi-27} in 1927. It was conceptually pushed forward by a variational principle of Hohenberg and Kohn \cite{HohKoh-64} in 1964, and since then many variations have been proposed.

In this paper, we are interested in the Levy-Lieb density functional \cite{Levy-79,Lieb-83}
\begin{align}\label{eq:FN}
\cL_N(\rho) = \inf \left\{ \langle \Psi_N, H_N \Psi_N \rangle: \Psi_N \in S_N , \rho_{\Psi_N}=\rho \right\}.
\end{align}
This is nicely related to the ground state  problem via the identity
\begin{align}\label{eq:EN-FN}
E_N^{\rm QM} = \inf \left\{ {\cL_N}(\rho): \rho \ge 0, \int_{\R^d} \rho =N \right\},
\end{align}
but we will consider \eqref{eq:FN} in a general context (without limiting to the ground state problem). The complication of the many-body problem is now hidden in the determination of ${\cL_N}$, and finding a good approximation is desirable. 

In this paper, we will focus on the typical situation when the particles are governed by the non-relativistic kinetic operator, an external potential and a pair-interaction potential, namely the Hamiltonian of the system reads (in  appropriate units)
\bq \label{eq:HN}
H_N=\sum_{i=1}^N \Big(- h^2\Delta_{x_i}+V(x_i)\Big)+\lambda \sum_{1\le i<j\le N} w(x_i-x_j).
\eq
Here $h>0$ plays the role of Planck's constant, and $\lambda>0$ corresponds to the strength of the interaction. The Thomas-Fermi approximation  \cite{Thomas-27,Fermi-27} suggests that 
\begin{align}\label{eq:FN}
{\cL_N}(\rho) \approx K_{\rm cl} h^2 \int_{\R^d} \rho^{1+2/d} + \int_{\R^d} V\rho + \frac{\lambda}{2}\iint_{\R^d \times \R^d} \rho(x)\rho(y) w(x-y) d x d y
\end{align}
where
$$
 K_{\rm cl}  = \frac{d}{d+2} \cdot  \frac{(2\pi)^2}{(q|B_{\R^d}(0,1)|)^{2/d}} .
$$

Historically, the Thomas-Fermi approximation was proposed for the atomic Hamiltonian, when $V$ and $w$ are Coulomb potentials in $\R^3$, but we may expect that it holds in a more general context. In this paper, we aim at giving rigorous justifications for \eqref{eq:FN} in the semiclassical mean-field regime 
$$
N\to \infty, \quad h \sim N^{-1/d}, \quad \lambda\sim N^{-1},
$$
which is natural to make all three terms on the right side of \eqref{eq:FN} comparable. 

To formulate our statements precisely, let us denote $\rho=Nf$ and rewrite \eqref{eq:FN} as 
\begin{align}\label{eq:TF-app-res}
{\cE_N}(f) \approx \cE^{\rm TF} (f)
\end{align}
where
\begin{align} \label{eq:LL-TF}
\cE_N(f) &= N^{-1}\cL_N(Nf ) = \inf \left\{ N^{-1} \langle \Psi_N, H_N \Psi_N \rangle: \Psi_N \in S_N, \rho_{\Psi_N}=Nf \right\},\nn\\
\cE^{\rm TF}(f)&= K_{\rm cl} \int_{\R^d} f^{1+2/d} + \int_{\R^d} Vf + \frac{1}{2}\iint_{\R^d \times \R^d} f(x)f(y) w(x-y) d x d y.
\end{align}

\subsection{An open problem} We expect that \eqref{eq:TF-app-res} holds for a very large class of potentials. However, to make the discussion concrete, let us assume the following conditions in the rest of the paper. 

\medskip
\noindent{\bf Conditions on potentials}. {\em The potentials $V,w:\R^d \to \R$ belong to $L^{p}(\R^d)+L^{q}(\R^d)$ with $p,q\in [1+d/2,\infty)$. 	Moreover, $w$ admits the decomposition
	\begin{align}\label{eq:Fd-decomp}
	w(x) = \int_0^{\infty} (\chi_r * \chi_r) (x)  dr,
	\end{align}
	for a  family of radial functions $0\le \chi_r \in L^{p}(\R^d)+L^{q}(\R^d)$ with  $p,q\in [2+d,\infty)$ }
\medskip

These assumptions hold for the Coulomb potentials in $\R^3$; in particular we have the Fefferman-de la Llave formula \cite{FefLla-86}
 $$
 \frac{1}{|x|} = \frac{1}{\pi}\int_0^{\infty} \frac{1}{r^5} (\1_{B_r} * \1_{B_r})(x)  dr
 $$
where $\1_{B_r}$ is the characteristic function of the ball $B(0,r)$ in $\R^3$. In fact, \eqref{eq:Fd-decomp} holds true for a large class of radial positive functions; see  \cite{HaiSei-02} for details.


Recall $\cE_N$ and $\cE^{\rm TF}$ in \eqref{eq:LL-TF}. We expect that the following holds true.

\begin{conjecture}[Semiclassical mean-field limit of Levy-Lieb functional] \label{conj-1} For all $d\ge 1$, in the limit $N\to \infty$, $hN^{1/d}\to 1$, $\lambda N \to 1$,  we have 
\begin{align} \label{eq:pointwise-cv}
\cE_N(f)\to \cE^{\rm TF}(f)
\end{align}
for every function $f$ satisfying $f\ge 0$, $\sqrt{f}\in H^1(\R^d)$ and $\int_{\R^d} f=1$. 
\end{conjecture}

Here are some immediate remarks on Conjecture \ref{conj-1}.

\medskip
\noindent 1) By the Hoffmann-Ostenhof inequality \cite{Hoffmann-Ostenhof-77}
\begin{align}
\left\langle \Psi_N, \sum_{i=1}^N (-\Delta_{x_i}) \Psi_N \right\rangle \ge \int_{\R^d} |\nabla \sqrt{\rho_{\Psi_N}}|^2,
\end{align}
the condition $\sqrt{f}\in H^1(\R^d)$ in Conjecture \ref{conj-1} is necessary to ensure that $\cE_{N}(f)<\infty$.

\medskip
\noindent 2) In the ideal Fermi gas (i.e. $V=w=0$), the Levy-Lieb functional boils down to the kinetic density functional 
\begin{align} \label{eq:def-cKN}
\cK_N(f):= \inf \left\{ N^{-1-2/d} \left\langle \Psi_N, \sum_{i=1}^N (-\Delta_{x_i}) \Psi_N \right\rangle: \Psi_N \in S_N,\rho_{\Psi_N}=Nf \right\}.
\end{align}
Conjecture \ref{conj-1} for the ideal Fermi gas states that, for all $d\ge 1$, 
\begin{align}\label{eq:conjecture-kinetic} \cK_N(f)\to K_{\rm cl} \int_{\R^d} f^{1+2/d}.
\end{align}
In fact, the following stronger, quantitative bounds are expected to hold \cite{MarYou-58,LieThi-75,Lieb-83}
\begin{align}\label{eq:conjecture-kinetic-stronger}  K_{\rm cl} \int_{\R^d} f^{1+2/d} \le \cK_N(f) \le K_{\rm cl} \int_{\R^d} f^{1+2/d} + N^{-2/d}\int_{\R^d} |\nabla \sqrt{f}|^2,
\end{align}
for all $d\ge 1$ for the upper bound and all $d\ge 3$ for the lower bound. 

The upper bound in \eqref{eq:conjecture-kinetic-stronger} was proposed by March and Young in 1958 \cite{MarYou-58}. Their proof works for $d=1$, but fails in higher dimensions (see \cite{Lieb-83} for a discussion, and see \cite{Acharya,GazRob} for numerical investigations in $d=3$). 

The lower bound in \eqref{eq:conjecture-kinetic-stronger} was conjectured by Lieb and Thirring in 1975 \cite{LieThi-75,LieThi-76} and they proved the bound with a universal constant (different from $K_{\rm cl}$). Note that in $d=1$, the sharp constant in the lower bound in \eqref{eq:conjecture-kinetic-stronger} is  known to be smaller than $K_{\rm cl}$ (it is conjectured to be the optimal constant in a Sobolev-Gagliardo-Nirenberg inequality \cite{LieThi-76}). Despite several improvements over the constant (see \cite{DolLapLos-08}  for the best known result), the sharp constant in the lower bound in \eqref{eq:conjecture-kinetic-stronger} is still open in all $d\ge 1$. On the other hand, recently we  proved that  \cite{Nam-18a}
\begin{align}\label{eq:LT-N} 
\cK_N(f)  \ge (K_{\rm cl}-\eps) \int_{\R^d} f^{1+2/d} - C_\eps N^{-2/d}\int_{\R^d} |\nabla \sqrt{f}|^2
\end{align}
for all $d\ge 1$ and all $\eps>0$. This implies the lower bound in \eqref{eq:conjecture-kinetic}. The upper bound in \eqref{eq:conjecture-kinetic}  is open for all $d\ge 2$.  

\medskip
\noindent 3) If we ignore the kinetic part, the Levy-Lieb functional reduces to the classical interaction functional 
\begin{align} \label{eq:IN-classical}
\mathcal{I}_N(f):=  \inf \left\{ N^{-2}\left\langle \Psi_N, \sum_{1\le i<j\le N} w(x_i-x_j) \Psi_N \right\rangle: \Psi_N \in S_N, \rho_{\Psi_N}=Nf \right\}.
\end{align}
Conjecture \ref{conj-1} becomes  
\begin{align} \label{eq:Con-int}
\lim_{N\to \infty} \mathcal{I}_N(f) = \frac{1}{2} \iint_{\R^d\times \R^d} f(x)f(y) w(x-y) d x d y,
\end{align}
which can be proved rigorously. In fact, when $w$ is the Coulomb potential in $\R^3$, the lower bound in \eqref{eq:Con-int} is a direct consequence of the Lieb-Oxford inequality \cite{LieOxf-81}
\begin{align} \label{eq:Lieb-Oxford}
\mathcal{I}_N(f) \ge \frac{1}{2} \iint_{\R^3\times \R^3} f(x)f(y) |x-y|^{-1} d x d y - 1.68 N^{-2/3} \int_{\R^3} f^{4/3}
\end{align}
and the upper bound in \eqref{eq:Con-int} can be achieved easily by choosing a Slater determinant  (a wave function of the form $\Psi_N=u_1 \wedge u_2 \wedge ... \wedge u_N$) with the density $Nf$. The proof of \eqref{eq:Con-int} for more general $w$ can be extracted from the proof of our main result below. Other approaches to \eqref{eq:Con-int} based on the optimal transportation have recently attracted a lot of attention, see \cite{ButtPascGior-12,CotFriKlu-13,FMPCK-13,CotFriPas-15,MarGer-2015}.

 If we do not completely ignore the kinetic part, but take $h\to 0$ and fix $N$, then the Levy-Lieb functional functional $\cE_N(f)$ converges to the interaction functional $\mathcal{I}_N(f)$ in \eqref{eq:IN-classical}. Results of this kind can be found in remarkable recent works \cite{CotFriKlu-13,BinPas-17, Lewin-18}. 
 
 The significance of Conjecture \ref{conj-1}, as well as our main result below, is the fact that we take the proper semiclassical limit $h\sim N^{-1/d}$ as $N\to \infty$, which is crucial to obtain the full Thomas-Fermi functional.

\subsection{Main result} While we could not prove the pointwise-type convergence in Conjecture \ref{conj-1}, we will provide another justification for the Thomas-Fermi functional from  the Levy-Lieb functional in the sense of the Gamma-Convergence.

Recall that $\cE_N$ and $\cE^{\rm TF}$ are defined in \eqref{eq:LL-TF}. We have

\begin{theorem} [Gamma convergence from Levy-Lieb to Thomas-Fermi functional]\label{thm:Gamma-CV2} For all $d\ge 1$, in the limit $N\to \infty$, $hN^{1/d}\to 1$, $\lambda N \to 1$, the Levy-Lieb functional $\cE_N$ Gamma-converges to the Thomas-Fermi functional $\cE^{\rm TF}$ in 
\begin{align} \label{eq:def-B}
\mathcal{B}=\left\{ 0\le f\in L^1(\R^d)\cap L^{1+2/d}(\R^d), \int_{\R^d} f=1 \right\}.
\end{align}
More precisely, we have
\begin{itemize}

\item[(i)] {\rm (Lower bound)} For every sequence $f_N \in \mathcal{B}$ such that $f_N \wto f$ weakly in $L^{1+2/d}(\R^d)$,  then
\begin{align} \label{eq:Gcv-lower}
\liminf_{N\to \infty}\cE_N(f_N) \ge \cE^{\rm TF}(f). 
\end{align}

\item[(ii)] {\rm  (Upper bound)} For every $f\in \mathcal{B}$, there exists a sequence of Slater determinants $\Psi_N\in S_N$ such that $f_N=N^{-1}\rho_{\Psi_N}\to f$ strongly in $L^1(\R^d)\cap L^{1+2/d}(\R^d)$, and 
\begin{align} \label{eq:Gcv-upper}
\limsup_{N\to\infty} \cE_N(f_N)\le\cE^{\rm TF}(f).
\end{align}
\end{itemize}
\end{theorem}

The notion of Gamma convergence is sufficient for many applications. In particular, we can come back to the ground state problem and immediately obtain

\begin{corollary}[Convergence of ground state energy and ground states] \label{cor:GSE}For all $d\ge 1$, in the limit $N\to \infty$, $hN^{1/d}\to 1$, $\lambda N \to 1$, the ground state energy $E_N^{\rm QM}$ of $H_N$ converges to the Thomas-Fermi energy:
\begin{align} \label{eq:cv-GSE}
\lim_{N\to \infty}\frac{E_N^{\rm QM}}{N}= E^{\rm TF} = \inf \left\{ \cE^{\rm TF}(f): f\in \mathcal{B} \right\}.
\end{align}
Moreover, if  $\Psi_N$ is a ground state for $E_N^{\rm QM}$ and if $f^{\rm TF}$ is a Thomas-Fermi minimizer, then  $N^{-1}\rho_{\Psi_N}\wto f^{\rm TF}$ weakly in $L^{1+2/d}(\R^d)$.
\end{corollary}

Corollary \ref{cor:GSE} covers the seminal result of Lieb and Simon \cite{LieSim-77} on the validity of Thomas-Fermi in the atomic case (when $V,w$ are Coulomb potentials in $\R^3$).

More general results on the ground state problem have been achieved recently by Fournais, Lewin and Solovej \cite{FouLewSol-15} by means of other techniques. Their method is based on a fermionic version of the de Finetti-Hewitt-Savage theorem for classical measures (which should be compared to a weak quantum de Finetti theorem for bosons \cite{LewNamRou-14}, although the classical version is sufficient for fermions). They can treat a very large class of interaction potentials; in particular no form of positivity, e.g. \eqref{eq:Fd-decomp}, is needed. In fact, negative interaction potentials can be handled by a clever technique of interchanging two-body to one-body potentials. The latter technique goes back to \cite{DysLen-67,LevLeb-69,LieThi-84,LieYau-87} and seems rather specific for the ground state problem.   

In contrast, our Theorem \ref{thm:Gamma-CV2} applies to a more restrictive class of interaction potentials, but it is not limited to ground states (it can be applied to excited, or high energy states as well).

The rest of the paper is devoted to the proof of our main result. First, we will study the ideal gas separately in Section \ref{sec:kinetic}. Then the proof of Theorem \ref{thm:Gamma-CV2} and Corollary \ref{cor:GSE} are given in Section \ref{sec:full}. 

In the following proof, we consider spinless particles for simplicity (adding a fixed spin $q\ge 1$ requires only straightforward modifications).

\section{Kinetic density functional} \label{sec:kinetic}

In this section we prove Theorem \ref{thm:Gamma-CV2} in the special case of the ideal Fermi gas, which has its own interest. Recall the kinetic density functional in \eqref{eq:def-cKN}
\begin{align*}
\cK_N(f)= \inf \left\{ \frac{1}{N^{1+2/d}} \left\langle \Psi_N, \sum_{i=1}^N (-\Delta_{x_i}) \Psi_N \right\rangle: \Psi_N \in S_N,\rho_{\Psi_N}=Nf \right\}
\end{align*}
and 
$$
\mathcal{B}=\left\{ 0\le f\in L^1(\R^d)\cap L^{1+2/d}(\R^d), \int_{\R^d} f=1 \right\}.
$$

We will prove

\begin{theorem} [Gamma convergence of kinetic density functional]\label{thm:Gamma-CV1} For all $d\ge 1$, the following convergences hold when $N\to \infty$. 

\begin{itemize}
\item[(i)] {\rm (Lower bound)} If $f_N \in \mathcal{B}$ and $f_N \wto f$ weakly in $L^{1+2/d}(\R^d)$,  then
\begin{equation} \label{eq:kinetic-lower}
\liminf_{N\to \infty}\cK_N(f_N) \ge K_{\rm cl} \int_{\R^d} f^{1+2/d}. 
\end{equation}

\item[(ii)] {\rm  (Upper bound)} For every $f\in \mathcal{B}$, there exists a sequence of Slater determinants $\Psi_N\in S_N$ such that $f_N=N^{-1}\rho_{\Psi_N}\to f$ strongly in $L^1(\R^d)\cap L^{1+2/d}(\R^d)$, and 
\begin{equation} \label{eq:kinetic-upper}
\limsup_{N\to\infty} \cK_N(f_N) \le \limsup_{N\to\infty} \frac{1}{N^{1+2/d}} \left\langle \Psi_N, \sum_{i=1}^N -\Delta_{x_i} \Psi_N \right\rangle  \le K_{\rm cl} \int_{\R^d} f^{1+2/d}.
\end{equation}
\end{itemize}
\end{theorem}

\begin{proof} {\bf Lower bound.} The lower bound \eqref{eq:kinetic-lower} is a consequence of Weyl's law for Schr\"odinger eigenvalues. Let $\Psi_N$ be a $N$-body wave function with density $\rho_{\Psi_N}=Nf_N$. We can define the one-body density matrix $\gamma_{\Psi_N}$ as a trace class operator on $L^2(\R^d)$ with kernel
\begin{align} \label{eq:def-1pdm}
\gamma_{\Psi_N}^{(1)}(x,y)= \int_{(\R^d)^{N-1}} \Psi_N(x,x_2,...,x_N) \overline{\Psi_N (y,x_2,...,x_N)} d x_2... dx_N. 
\end{align}
Then for every function $0\le U\in C_c^\infty(\R^d)$, we can write
\begin{align} \label{eq:lower-easy}
h^2 \left\langle \Psi_N, \sum_{i=1}^N -\Delta_{x_i} \Psi_N \right\rangle = \Tr\left[(-h^2 \Delta - U)\gamma_{\Psi_N}^{(1)}\right] + N \int_{\R^d} U f_N.
\end{align}
On the other hand, the anti-symmetry of $\Psi_N$ implies Pauli's exclusion principle \cite[Theorem 3.2]{LieSei-10}
\begin{align} \label{eq:Pauli}
0\le \gamma_{\Psi_N}^{(1)} \le 1.
\end{align}
Consequently, by the min-max principle  \cite[Theorem 12.1]{LieLos-01} and Weyl's law on the sum of negative eigenvalues of Schr\"odinger operators \cite[Theorem 12.12]{LieLos-01} we can estimate
\begin{align} \label{eq:Weyl}
\Tr\left[(-N^{-2/d} \Delta - U)\gamma_{\Psi_N}^{(1)}\right] &\ge \Tr[-N^{-2/d} \Delta - U]_-\\ 
&= - N\frac{|B_{\R^d}(0,1)|}{(2\pi)^d(1+d/2)} \left[ \int_{\R^d} U^{1+d/2} + o(1)_{N\to \infty}\right]. \nn
\end{align}
From \eqref{eq:lower-easy} and \eqref{eq:Weyl}  we deduce that
$$
\liminf_{N\to \infty} \cK_N(f_N) \ge - \frac{|B_{\R^d}(0,1)|}{(2\pi)^d(1+d/2)} \int_{\R^d} U^{1+d/2} + \int_{\R^d} U f.
$$
Optimizing over $U$ leads to the desired lower bound \eqref{eq:kinetic-lower}. 

\bigskip

\noindent {\bf Upper bound.} We can follow the coherent state approach in the proof of Weyl's law \cite[Theorem 12.12]{LieLos-01} to deduce the upper bound \eqref{eq:kinetic-upper}, but the important conclusion that the density $Nf_N$ comes from a Slater determinant is not easily achieved in this way. In the following, we will provide a direct proof of the upper bound in Theorem \ref{thm:Gamma-CV1}, using explicit computations of the ground states of the ideal Fermi gas in cubes. This idea goes back to the heuristic argument of Thomas-Fermi  \cite{Thomas-27,Fermi-27} and March-Young \cite{MarYou-58} (the same argument can be used to give a direct proof of the lower bound \eqref{eq:kinetic-lower}; see \cite{Gottschling-18} for details). 
 
\medskip 

\noindent
 \textbf{Step 1 (Slater determinants with step-function densities).}  Recall that that the Dirichlet Laplacian $-\Delta$ on the cube $Q=[0,L]^d$ has eigenvalues $|\pi k/L|^2$, $k\in \mathbb{N}^d$, with eigenfunctions
 $$
u_k(x)= \prod_{i=1}^d \left[ \sqrt{\frac{2}{L}}\sin\bigg(\frac{\pi k^i x^i}{L}\bigg) \right], \quad k=(k^i)_{i=1}^d, x=(x^i)_{i=1}^d \in \R^d.
$$
The ground state of the $M$-body kinetic operator $\sum_{j=1}^M (-\Delta_{x_j})$ is the Slater determinant $\Psi_M^{\rm S}$ made of the first $M$ eigenfunctions $\{u_k\}$. It is straightforward to see that when $M\to \infty$,
\begin{align}
		 \frac{1}{M}\rho_{\Psi_M^{S}}= \frac{1}{M}\sum_{k\in S_M} |u_k|^2 \to \frac{\1_Q}{|Q|} 
		 \label{eq:simplecube1}
		\end{align}
				strongly in $L^p(Q)$ for all $1\leq p<\infty$, and
				\begin{align}\label{eq:simplecube2}
		\frac{1}{M^{1+2/d}} \left\langle \Psi_M^{S},\sum_{i=1}^M(-\Delta_{x_i}) \Psi_M^{S}\right\rangle = \frac{1}{M^{1+2/d}} \sum_{k\in S_M} \left|\frac{\pi k}{L}\right|^2  \to \frac{K_{cl}}{|Q|^{2/d}}.
		\end{align}

Now let $f\in \mathcal{B}$. Let $\{Q\}$ be a finite family of disjoint cubes, whose construction will be specified in the next step. In the following we only consider cubes $Q$ such that 
$$\overline{f}^Q := \frac{1}{|Q|}\int_Q f>0,.$$ We can find an integer number
$M_Q\in (N |Q| \overline{f}^Q -1, N|Q| \overline{f}^Q+1]$
such that 
$$\sum_{Q}M_Q= \sum_{Q} N |Q| \overline{f}^Q =  N \int_{\R^d} f=1.$$
	
Now for every $Q$, consider the first $M_Q$ eigenfunctions $\{u_j^Q\}_{j=1}^{M_Q}$ of the Dirichlet Laplacian $-\Delta$ on $Q$. These functions can be trivially extended to zero outside $Q$ to become a function in $H^1_0(Q_0)$. Since $\{u_j^Q\}_{j=1}^{M_Q}$ are orthogonal for every $Q$ and the subcubes $\{Q\}$ are disjoint, the collection $\bigcup_{Q} \{u_j^Q\}_{j=1}^{M_Q}$ is an orthonormal family of $N$ functions in $H^1_0(Q_0)$. Let $\Psi_N^{S}\in S_N$ be the Slater determinant made of this orthogonal family. Then in the limit $N\to \infty$, using the fact that $M_Q/N\to |Q| \overline{f}^Q >0$ and  \eqref{eq:simplecube1}, \eqref{eq:simplecube2}, we get the following 
		
		\begin{align}\label{eq:cube1}
		\frac{1}{N}\rho_{\Psi_N^S} =  \sum_Q \sum_{i=1}^{M_Q}  \frac{|u_i^Q|^2}{M_Q} \cdot \frac{M_Q}{N} & \to \sum_Q \frac{\1_Q}{|Q|} \cdot |Q|\overline{f}^Q = \sum_Q \1_Q \overline{f}^Q 
		\end{align}
		strongly in $L^p(\R^d)$ for all $1\leq p<\infty$, and
		\begin{align}\label{eq:cube2}
		&\frac{1}{N^{1+2/d}} \left\langle \Psi_N^{S}, \sum_{i=1}^N - \Delta_{x_i} \Psi_N^{S} \right\rangle  = \frac{1}{N^{1+2/d}} \sum_Q \sum_{i=1}^{M_Q}\lVert\nabla u_i^Q\rVert^2 \nn\\
		&= \sum_Q \left[ \frac{1}{M_Q^{1+2/d}} \sum_{i=1}^{M_Q}\lVert\nabla u_i^Q\rVert^2 \right] \left|\frac{M_Q}{N}\right|^{1+2/d} \to \sum_Q \frac{K_{\rm cl}}{|Q|^{2/d}} \cdot \left| |Q| \overline{f}^Q \right|^{1+2/d} \nn\\
		&= K_{\rm cl} \int_{\R^d} \Big| \sum_Q \1_Q \overline{f}^Q \Big|^{1+2/d} \le K_{\rm cl} \int_{\R^d} f^{1+2/d} . 
		\end{align}
\noindent		
\textbf{Step 2 (Approximating $f$ by step functions and concluding).} Since $0\le f\in L^1(\R^d)\cap L^{1+2/d}(\R^d)$, for every $k \ge 1$ we can find a finite family of disjoint cubes $\{Q\}$ such that  
\begin{align} \label{eq:1Q-f<e}
\left\|f - \sum_{Q} \1_Q \overline{f}^Q \right\|_{L^1} + \left\|f - \sum_{Q} \1_Q \overline{f}^Q \right\|_{L^{1+2/d}} \le k^{-1}.  
\end{align}
Using this collection of cubes, for every $N\ge 1$ we can construct a Slater determinant $\Psi_N^\eps \in S_N$ as in Step 2. From the convergence \eqref{eq:cube1}, \eqref{eq:cube2}, we deduce that there exists $M_k >0$ such that for every $N\ge M_k$, 
\begin{align} \label{eq:kin-upp-kin}
\frac{1}{N^{1+2/d}} \left\langle \Psi_N^k, \sum_{i=1}^N - \Delta_{x_i} \Psi_N^k \right\rangle \le K_{\rm cl} \int_{\R^d} f^{1+2/d} + k^{-1}
\end{align}
and
$$
\left\| \frac{1}{N}\rho_{\Psi_N^k} - \sum_Q \1_Q \overline{f}^Q  \right\|_{L^1} + \left\| \frac{1}{N}\rho_{\Psi_N^k} - \sum_Q \1_Q \overline{f}^Q  \right\|_{L^{1+2/d}} \le k^{-1}.
$$
The latter estimate and \eqref{eq:1Q-f<e} imply that for every $N\ge M_k$,
\begin{align} \label{eq:kin-upp-den}
\left\| \frac{1}{N}\rho_{\Psi_N^k} - f  \right\|_{L^1} +  \left\| \frac{1}{N}\rho_{\Psi_N^k} - f  \right\|_{L^{1+2/d}}   \le 2 k^{-1}. 
\end{align}

Now we conclude using a standard diagonal argument. By induction we can choose the above sequence $M_k$ such that $M_{k+1}\ge M_k+1$. Now for every $N\in \mathbb{N}$, we take $k=k_N$ the smallest number such that $N\ge M_k$. Obviously we have $k_N\to \infty$ as $N\to \infty$. Moreover, we can choose the Slater determinant $\Psi_N=\Psi_N^{k_N}\in S_N$ as above, and obtain from \eqref{eq:kin-upp-kin}, \eqref{eq:kin-upp-den} that 
$$
\frac{1}{N^{1+2/d}} \left\langle \Psi_N, \sum_{i=1}^N - \Delta_{x_i} \Psi_N \right\rangle \le K_{\rm cl} \int_{\R^d} f^{1+2/d} + k_N^{-1} \to K_{\rm cl} \int_{\R^d} f^{1+2/d} 
$$
and
$$
\left\| \frac{1}{N}\rho_{\Psi_N} -f \right\|_{L^1} + \left\| \frac{1}{N}\rho_{\Psi_N} - f \right\|_{L^{1+2/d}} \le 2k_N^{-1} \to 0
$$
when $N\to \infty$. This completes the proof of Theorem \ref{thm:Gamma-CV1}. 
\end{proof}

\section{Full density functional} \label{sec:full}

In this section we prove our main result. 

\begin{proof}[Proof of Theorem \ref{thm:Gamma-CV2}]  {\bf Lower bound.}  Let $\Psi_N\in S_N$ such that $\rho_{\Psi_N}=Nf_N$ and $f_N\wto f$ weakly in $L^{1+2/d}(\R^d)$.  By Theorem \ref{thm:Gamma-CV1}, we have
\begin{align} \label{eq:proof-main-1}
\liminf_{N\to \infty} N^{-1-2/d} N^{2/d}h^2 \left\langle \Psi_N, \sum_{i=1}^N -\Delta_{x_i} \Psi_N \right\rangle \ge K_{\rm cl} \int_{\R^d} f^{1+2/d}. 
\end{align}

Moreover, since $f_N\wto f$ weakly in $L^{1+2/d}(\R^d)$ and $\|f_N\|_{L^1}=1$, by interpolation we have  $f_N\wto f$ weakly in $L^{r}(\R^d)$ for all $r \in (1, 1+2/d]$. Under the condition $V\in L^p(\R^d)+L^q(\R^d)$ with $p,q\in [1+d/2,\infty)$, we deduce that
\begin{align} \label{eq:proof-main-2}
\lim_{N\to \infty} N^{-1} \left\langle \Psi_N, \sum_{i=1}^N V(x_i) \Psi_N \right\rangle =  \lim_{N\to \infty} \int_{\R^d} Vf_N = \int_{\R^d} Vf. 
\end{align}

It remains to consider the interaction terms.  We will use an idea of Lieb, Solovej and Yngvason \cite{LieSolYv-95}, 
which has been used to give an alternative proof of the Lieb-Oxford inequality. From the Fefferman-de la Llave type decomposition \eqref{eq:Fd-decomp}, we can write
\begin{align} \label{eq:FeffDela-xy}
	w(x-y) = \int_0^{\infty} dr   \int_{\R^d} dz \chi_r(x-z) \chi_r(y-z) dz 
\end{align}
	and hence 
	\begin{align*}
	\left\langle \Psi_N, \sum_{1\le i<j\le N}w(x_i-x_j)\Psi_N \right\rangle= 
	 \int_0^{\infty} dr \int_{\R^d} dz \left  \langle  \Psi_N, \sum_{1\le i<j\le N}\chi_r(x_i-z) \chi_r(x_j-z)\Psi_N \right\rangle .
	\end{align*}
By the Cauchy-Schwarz inequality we get 
	\begin{align*}
	 &\left  \langle  \Psi_N, \sum_{1\le i<j\le N}\chi_r(x_i-z) \chi_r(x_j-z)\Psi_N \right\rangle \\
	 &= \left[ \frac{1}{2} \left  \langle  \Psi_N, \Big( \sum_{i=1}^N \chi_r(x_i-z) \Big)^2 \Psi_N \right\rangle - \left  \langle  \Psi_N, \sum_{i=1}^N \chi_r^2(x_i-z)  \Psi_N \right\rangle \right]_+\\
	 &\ge \left[ \frac{1}{2} \left  \langle  \Psi_N, \sum_{i=1}^N \chi_r(x_i-z)  \Psi_N \right\rangle^2 - \left  \langle  \Psi_N, \sum_{i=1}^N \chi_r^2(x_i-z)  \Psi_N \right\rangle \right]_+\\
	 &= \left[ \frac{N^2}{2}(f_N*\chi_r)^2(z) - N (f_N*\chi_r^2)(z) \right]_+.
	\end{align*}
	For every fixed $r>0$ and $z\in \mathbb{\R^d}$, since $f_N\wto f$ weakly in $L^{r}(\R^d)$ for all $1<r\le 1+2/d$, and $\chi_r,\chi_r^2 \in L^p(\R^d)+L^q(\R^d)$ with $p,q\in [1+d/2,\infty)$, we find that
	\begin{align*}
	\lim_{N\to \infty} (f_N*\chi_r)(z)  &= (f*\chi_r)(z), \\
	\lim_{N\to \infty} (f_N*\chi_r^2)(z)  &= (f*\chi_r^2)(z),
	\end{align*}
	and hence
	$$
	\lim_{N\to \infty} N^{-2} \lambda N \left[ \frac{N^2}{2}(f_N*\chi_r)^2(z) - N (f_N*\chi_r^2)(z) \right]_+ = \frac{1}{2}(f*\chi_r)^2(z)
	$$
	for every $z\in \R^d$. Therefore, by Fatou's lemma,
	\begin{align} \label{eq:proof-main-3}
	&\liminf_{N\to \infty} N^{-2} \lambda N  \left\langle \Psi_N, \sum_{1\le i<j\le N}w(x_i-x_j)\Psi_N \right\rangle\nn\\
	 &=\liminf_{N\to \infty} \int_0^{\infty} dr \int_{\R^d} dz  N^{-2} \lambda N \left[ \frac{N^2}{2}(f_N*\chi_r)^2(z) - N (f_N*\chi_r^2)(z) \right]_+ \nn\\
	 &\ge \int_0^{\infty} dr \int_{\R^d} dz \frac{1}{2}(f*\chi_r)^2(z)= \frac{1}{2} \iint_{\R^d\times \R^d} f(x)f(y) w(x-y) d x dy. 
	\end{align}
	Here in the last identity we have used \eqref{eq:Fd-decomp} again.

	Putting \eqref{eq:proof-main-1}, \eqref{eq:proof-main-2} and \eqref{eq:proof-main-3} together, we conclude that
	$$
	\liminf_{N\to \infty} N^{-1}\langle \Psi_N, H_N \Psi_N\rangle \ge \cE^{\rm TF}(f). 
	$$
	Since $\Psi_N\in S_N$ can be chosen arbitrarily under the sole condition $\rho_{\Psi_N}=Nf_N$, this leads the desired lower bond
	$$
	\liminf_{N\to \infty} \cE_N (f_N) \ge \cE^{\rm TF}(f). 
	$$
\noindent{\bf Upper bound.} Let $0\le f\in L^1(\R^d)\cap L^{1+2/d}(\R^d)$ with $\int_{\R^d} f=1$. Then by Theorem \ref{thm:Gamma-CV1} there exists a sequence of Slater determinants $\Psi_N\in S_N$, such that 
	$f_N:= N^{-1}\rho_{\Psi_N} \to f$ strongly in $L^1(\R^d)\cap L^{1+2/d}(\R^d)$ and 
	\begin{align}\label{eq:proof-main-upper-1}
	\limsup_{N\to \infty} N^{-1-2/d} N^{2/d}h^2  \left\langle \Psi_N, \sum_{i=1}^N (-\Delta_{x_i}) \Psi_N \right\rangle \le K_{\rm cl}\int_{\R^d} f^{1+2/d}. 
	\end{align}
	
	Since $f_N\to f$ in $L^r(\R^d)$ for all $r\in [1,1+2/d]$ and $V\in L^p(\R^d)+L^q(\R^d)$ with $p,q\in [1+d/2,\infty)$, we have
	\begin{align} \label{eq:proof-main-upper-2}
\lim_{N\to \infty} N^{-1} \left\langle \Psi_N, \sum_{i=1}^N V(x_i) \Psi_N \right\rangle =  \lim_{N\to \infty} \int_{\R^d} Vf_N = \int_{\R^d} Vf. 
\end{align}
	
	Finally, for the interaction terms, since $\Psi_N$ is a Slater determinants and $w$ is non-negative, an explicit computation shows that 
	\begin{align} \label{eq:proof-main-upper-3}
	N^{-2} \lambda N \left\langle \Psi_N, \sum_{1\le i<j\le N}w(x_i-x_j)\Psi_N \right\rangle\le \frac{1}{2} \lambda N \iint_{\R^d\times \R^d} f_N(x)f_N(y) w(x-y) d x dy. 
	\end{align}
	Here since $w\ge 0$ we can simply ignored the exchange term in the Hartree-Fock functional to get an upper bound (see e.g. \cite[Section 5A]{Lieb-83} for details). The convergence $f_N\to f$ in $L^1(\R^d)\cap L^{1+2/d}(\R^d)$ and the assumption $w\in L^p(\R^d)+L^q(\R^d)$ imply that $f_N*w\to f*w$ strongly in $L^\infty(\R^d)$, and hence
	\begin{align*}
	&\frac{1}{2} \lambda N \iint_{\R^d\times \R^d} f_N(x)f_N(y) w(x-y) d x dy = \frac{1}{2} \lambda N \int_{\R^d} f_N (f_N*w) \\
	&\to  \frac{1}{2} \int_{\R^d} f (f*w) =  \frac{1}{2} \iint_{\R^d\times \R^d} f(x)f(y) w(x-y) d x dy.
	\end{align*}
	Putting this together with \eqref{eq:proof-main-upper-1}, \eqref{eq:proof-main-upper-2} and \eqref{eq:proof-main-upper-3} we obtain the desired upper bound
	$$
	\limsup_{N\to \infty} N^{-1}\langle \Psi_N, H_N \Psi_N\rangle \le \cE^{\rm TF}(f),
	$$
\end{proof}

\begin{proof}[Proof of Corollary \ref{cor:GSE}] The upper bound in \eqref{eq:cv-GSE}, $N^{-1} E_N^{\rm QM} \le E^{\rm TF}+o(1)_{N\to \infty}$, follows immediately from Theorem \ref{thm:Gamma-CV2} (upper bound) and optimizing over $f$ in \eqref{eq:Gcv-upper}. 

To see the lower bound in \eqref{eq:cv-GSE}, we take arbitrarily a $N$-body wave function $\Psi_N$ such that 
\begin{align} \label{eq:approx-GS}
N^{-1}\langle \Psi_N, H_N \Psi_N\rangle = N^{-1} E_N^{\rm QM}+ o(1)_{N\to \infty} \le E^{\rm TF}+ o(1)_{N\to \infty}. 
\end{align}
Denote $\rho_{\Psi_N}=Nf_N$. Using $w\ge 0$, the Lieb-Thirring inequality for the kinetic energy \cite{LieThi-76}, H\"older's inequality and the assumption $V \in L^{p}(\R^d)+L^{q}(\R^d)$ with $p,q\in [1+d/2,\infty)$  we can estimate
\begin{align*}
N^{-1}\langle \Psi_N, H_N \Psi_N\rangle &\ge N^{-1}\langle \Psi_N, \sum_{i=1}^N (-\Delta_{x_i}+V(x_i)) \Psi_N\rangle \nn\\
& \ge   K \int_{\R^d} f_N^{1+2/d} - \int_{\R^d} Vf_N \ge (K/2) \int_{\R^d} f_N^{1+2/d} - C
\end{align*}
where $K,C>0$ are constants independent of $f_N$. Thus from \eqref{eq:approx-GS} deduce that $f_N$ is bounded in $L^{1+2/d}(\R^d)$. 

Up to a subsequence, $f_N \wto f$ in $L^{1+2/d}(\R^d)$, and hence Theorem \ref{thm:Gamma-CV2} (lower bound) implies that
\begin{align}\label{eq:lower-f-<=1}
N^{-1}\langle \Psi_N, H_N \Psi_N\rangle \ge \cE_N(f_N) \ge \cE^{\rm TF}(f) + o(1)_{N\to \infty}. 
\end{align}

Next, let us show that 
\begin{align}\label{eq:cE-f-cE-TF} 
E^{\rm TF}= \inf\left\{ \cE^{\rm TF}(g): 0\le g \in L^1(\R^d)\cap L^{1+2/d}(\R^d), \int_{\R^d} g \le 1\right\}.
\end{align}
This follows from a standard argument.  If $\int_{\R^d} g\le 1$, we can take a function 
$$0\le \varphi \in C_c^\infty(\R^d), \quad \int_{\R^d} \varphi+ \int_{\R^d}g =1.$$
Take a sequence $\{R_k\}\subset \R^d, |R_k|\to \infty$. By the variational principle
\begin{align*}
E^{\rm TF} &\le \lim_{k\to \infty} \cE^{\rm TF} (g+ \varphi(\cdot+R_k)) \nn\\
&= \cE^{\rm TF}(g) + K_{\rm cl} \int_{\R^d} |\varphi|^{1+2/d} + \frac{1}{2}\iint_{\R^d} \varphi(x)\varphi(y) w(x-y) dx dy \nn\\
&\le \cE^{\rm TF}(g) + K_{\rm cl} \int_{\R^d} |\varphi|^{1+2/d} + C( \|\varphi\|_{L^r}^2+ \|\varphi\|_{L^{s}}^2).
\end{align*}
In the last estimate we have used Young's inequality \cite[Theorem 4.2]{LieLos-01} and the assumption $w \in L^{p}(\R^d)+L^{q}(\R^d)$ with $p,q\in [1+d/2,\infty)$. Here the parameters $r,s>1$ are determined by
$$
\frac{1}{p}+\frac{2}{r}=2 = \frac{1}{q}+\frac{2}{s}
$$
and the constant $C>0$ depends only on $w$. By scaling $\varphi\mapsto \ell^{d}\varphi(\ell \cdot)$ with $\ell\to 0$, we conclude that $E^{\rm TF} \le \cE^{\rm TF}(g)$. Thus \eqref{eq:cE-f-cE-TF} holds. 

Note that the weak convergence $f_N \wto f$ implies that $\int_{\R^d} f \le 1$. Therefore, combining \eqref{eq:lower-f-<=1} and \eqref{eq:cE-f-cE-TF} we arrive at
$$
N^{-1}\langle \Psi_N, H_N \Psi_N\rangle \ge \cE_N(f_N) \ge \cE^{\rm TF}(f) + o(1)_{N\to \infty} \ge E^{\rm TF}+ o(1)_{N\to \infty}. 
$$
Thanks to \eqref{eq:approx-GS}, we obtain the convergence \eqref{eq:cv-GSE} and that $\cE^{\rm TF}(f)=E^{\rm TF}$. 

Finally, note that $\cE^{\rm TF}(g)$ is strictly convex in $g$. This can be seen from the strict convexity of the kinetic term $g\mapsto K_{\rm cl}\int_{\R^d} g^{1+2/d}$ and the convexity of the interaction term
$$
g\mapsto \frac{1}{2}\iint_{\R^d\times \R^d} g(x) g(y)w(x-y) d x dy= \frac{1}{2}\int_0^\infty \left[  \int_{\R^d} |(g*\chi_t)(z)|^2 dz \right] dt. 
$$
Here we have used again the Fefferman-de la Llave formula \eqref{eq:Fd-decomp}. Consequently, if $E^{\rm TF}$ has a minimizer $f^{\rm TF}$, then using $\cE^{\rm TF}(f)=E^{\rm TF}=\cE^{\rm TF}(f^{\rm TF})$, the strict convexity and \eqref{eq:cE-f-cE-TF}, we conclude that $f=f^{\rm TF}$. Thus $f_N = N^{-1}\rho_{\Psi_N} \wto f^{\rm TF}$ weakly in $L^{1+2/d}(\R^d)$, for every wave function $\Psi_N$ satisfying \eqref{eq:approx-GS} (not necessarily a ground state of $H_N$). 
\end{proof}


\end{document}